\documentclass{article}
\usepackage[utf8]{inputenc}
\usepackage{array}
\usepackage{wrapfig}
\usepackage{multirow}
\usepackage{tabu}
\usepackage{ dsfont }
\usepackage{graphics} 
\usepackage{epsfig} 
\usepackage{mathptmx} 
\usepackage{times} 
\usepackage{amsmath} 
\usepackage{amssymb}  
\usepackage{amsthm}
\usepackage{hyperref}
\usepackage{authblk}

\newcommand*{\email}[1]{%
    \normalsize\href{mailto:#1}{#1}\par
}

\newtheorem{thm}{Theorem}[section]
\newtheorem{lem}[thm]{Lemma}
\newtheorem{prop}[thm]{Proposition}
\newtheorem{cor}[thm]{Corollary}

\begin{document}
\title{Complexity of Partitioning Hypergraphs }
\author[1]{Seonghyuk Im }
\affil[1]{Department of Mathematical Sciences, KAIST  \email{cblimmm@kaist.ac.kr}}

\maketitle
\begin{abstract}
    For a given $\pi=(\pi_0, \pi_1,..., \pi_k) \in \{0, 1, *\}^{k+1}$, we want to determine whether an input $k$-uniform hypergraph $G=(V, E)$ has a partition $(V_1, V_2)$ of the vertex set so that for all $X \subseteq V$ of size $k$, $X \in E$ if $\pi_{|X\cap V_1|}=1$ and $X \notin E$ if $\pi_{|X\cap V_1|}=0$. We prove that this problem is either polynomial-time solvable or NP-complete depending on $\pi$ when $k=3$ or $4$. We also extend this result into $k$-uniform hypergraphs for $k \geq 5$.
\end{abstract}
\section{Introduction}\label{sec1}

For a given $\pi=(\pi_0, \pi_1,..., \pi_k) \in \{0, 1, *\}^{k+1}$, the $\pi$-\textsc{partition} is a decision problem with the following input and output.

\begin{center}
\begin{tabular}{ | m{5em} | m{9cm}| } 
 \hline
 \multicolumn{2}{|c|}{$\pi$-\textsc{partition}} \\ 
 \hline
 input & A $k$-uniform hypergraph $G=(V,E)$  \\ 
 \hline
 output & True if $G$ has a partition $(V_1, V_2)$ of the vertex set so that for all $X \subseteq V$ of size $k$, $X \in E$ if $\pi_{|X\cap V_1|}=1$ and $X \notin E$ if $\pi_{|X\cap V_1|}=0$. False otherwise. \\ 
 \hline
\end{tabular}
\end{center}
We say such partition as a \emph{$\pi$-partition} and a $k$-uniform hypergraph $G$ is \emph{$\pi$-partitionable} if it has a $\pi$-partition.

Some special cases of $\pi$-\textsc{partition} have been studied by many people. Table \ref{table1} shows a list of some cases of $\pi$-\textsc{partition} which has known time complexity.
\begin{table}[h!]
\centering
\begin{tabular}{ | m{6em} | m{4cm}| m{4.5cm} | } 
 \hline
 $\pi$ & $\pi$-partitionable graph & Complexity to decide it \\
 \hline
 $(0, *, 0)$ & 2-colorable graph & Linear \\
 \hline
 $(0, 1, 0)$ & Complete bipartite graph & Linear \\
 \hline
 $(1, *, 0)$ & Split graph & Linear \cite{zbMATH03773637} \\
 \hline
 $(0, *, *, 0)$ & $2$-colorable $3$-uniform hypergraph & NP-complete \cite{MR0363980} \\
 \hline
 $\pi_0=0, \pi_k=1$ and $\pi_\ell=*$ otherwise & Split hypergraph & Polynomial-time solvable \cite{MR2002173} \cite{MR2424844} \\
 \hline
\end{tabular}
\caption{List of $\pi$-\textsc{partition}s with known complexity}
\label{table1}
\end{table}

The CSP dichotomy conjecture implies that the $\pi$-\textsc{partition} is polynomial-time solvable or NP-complete when $\pi$ does not contains $1$. Recently, Zhuk \cite{MR3734241} claims the proof of this conjecture . We will completely classify $\pi$-\textsc{partition} into polynomial-time solvable problems and NP-complete problems for $k=3, 4$ and we will prove some partial result for $k \geq 5$.

Before introducing our results, we observe some trivial facts. First, there exist a trivial partition if $\pi_0=*$ or $\pi_k=*$. For a $\pi \in \{0, 1, *\}^{k+1}$, let $\pi'$ be a vector in $ \{0, 1, *\}^{k+1}$ such that every $1$ of $\pi$ changed into $0$ and every $0$ of $\pi$ changed into $1$. Then by taking complement, it is easy to see that the $\pi$-\textsc{partition} and the $\pi'$-\textsc{partition} are polynomial time equivalent. Therefore, we may assume that $\pi_0=0$ for all $\pi$. Similarly, by changing label of $V_1$ and $V_2$, we get the following fact.
\begin{lem}\label{lem1}
The $(\pi_0, \pi_1, ...,\pi_k)$-\textsc{partition} and the $(\pi_k, \pi_{k-1}, ...,\pi_0)$-\textsc{partition} are polynomial-time equivalent.
\end{lem}

When $k=2$, $(0, *, 0)$-\textsc{partition}, $(0, 1, 0)$-\textsc{partition} and $(0, 0, 1)$-\textsc{partition} are polynomial-time solvable by using DFS(Depth First Search) and $(0, 0, 0)$-\textsc{partition} is clearly polynomial-time solvable. $(0, *, 1)$-\textsc{partition} is equivalent to deciding whether a graph is a split graph. Hammer and Simeone \cite{zbMATH03773637} proved it is polynomial-time solvable. By considering trivial cases and polynomial-time equivalences described above, we get the $\pi$-\textsc{partition} is polynomial-time solvable for all $\pi \in \{0, 1, *\}^3$. 

We proved full dichotomy result when $k=3$ or $4$ and we proved some partial dichotomy result for larger $k$. 
\begin{thm}\label{thm3_uni}
For a $\pi \in \{0, 1, *\}^{4}$, the $\pi$-\textsc{partition} is NP-complete if $\pi=(0, *, *, 0),$ $\pi=(1, *, *, 1),$ $\pi=(0, 0, *, 0),$ $\pi=(0, *, 0, 0)$, $\pi=(1, *, 1, 1)$ or $\pi=(1, 1, *, 1)$ and polynomial-time solvable otherwise.
\end{thm}
\begin{thm}\label{thm4_uni}
For a $\pi \in \{0, 1, *\}^{5}$, the $\pi$-\textsc{partition} is NP-complete if $\pi=(0, *, *, *, 0)$, $\pi=(0, *, *, 0, 0)$, $\pi=(0, 0, *, *, 0)$, $\pi=(0, *, 0, 0, 0)$, $\pi=(0, 0, 0, *, 0)$, $\pi=(0, 0, *, 0, 0)$,  $\pi=(1, *, *, *, 1)$, $\pi=(1, *, *, 1, 1)$, $\pi=(1, 1, *, *, 1)$, $\pi=(1, *, 1, 1, 1)$, $\pi=(1, 1, 1, *, 1)$ or $\pi=(1, 1, *, 1, 1)$ and polynomial-time solvable otherwise.
\end{thm}
\begin{thm}\label{thm5_uni}
Let $\pi=(\pi_0, ..., \pi_5) \in \{0, 1, *\}^6$ be a vector. Suppose $\pi_0=0$, $\pi \neq (0, *, 0, 0, *, 0)$, $\pi \neq (0, *, *, 0, *, 0)$, $\pi \neq (0, *, 0, *, *, 0)$, $\pi \neq (0, *, 0, *, 0, 0)$ and $\pi \neq (0, 0, *, 0, *, 0)$. Then the $\pi$-\textsc{partition} is NP-complete if $\pi=(0, *, *, *, *, 0)$, $\pi=(0, *, *, *, 0, 0)$, $\pi=(0, 0, *, *, *, 0)$, $\pi=(0, 0, 0, *, *, 0)$, $\pi=(0, 0, *, *, 0, 0)$, $\pi=(0, *, *, 0, 0, 0)$, $\pi=(0, 0, 0, 0, *, 0)$, $\pi=(0, 0, 0, *, 0, 0)$, $\pi=(0, 0, *, 0, 0, 0)$ or $\pi=(0, *, 0, 0, 0, 0)$ and polynomial-time solvable otherwise.
\end{thm}
To show these theorems, we first prove the fact that if $\pi$ contains both $0$ and $1$, then the $\pi$-\textsc{partition} is polynomial-time solvable in Subsection \ref{sec22}. We also prove that if $\pi=(0, *, 0, *,..., *, 0)$, then the $\pi$-\textsc{partition} is polynomial-time solvable in the same subsection.
Then it is enough to check $(0, *, 0, 0)$-\textsc{partition} is NP-complete or polynomial-time solvable to get full dichotomy result of $k=3$ cases. We will show NP-completeness of this problem in Subsection \ref{sec23}. After that, we will look some polynomial-time reductions from a larger $k$ to a smaller $k$ so that we complete the proof of the second and third theorem using the cases of $k=3$ in Subsection \ref{sec24}. 
\section{Preliminary}
\subsection{Basic notations}
For a $k$-uniform hypergraph $G=(V, E)$ and a subset of vertices $V' \subseteq V$, a \textit{subgraph induced by $V'$} be a hypergraph $G'$ with the vertex set $V'$ and the edge set $E(G')=\{e \in E(G) \mid e \subseteq V' \}$ and denoted by $G[V']$. If $G[V']$ is an empty graph, we say $V'$ is an \textit{independent set} of $G$. Similarly, if $G[V']$ has edges for every $k$-subset of $V'$, we say $V'$ is a \textit{clique} of $G$.

Let $C_m^k$ be a $k$-uniform hypergraph with a vertex set $\mathds{Z}_m$ and an edge set $\{\{ i, i+1, ..., i+k-1\} \mid 1 \leq i \leq m \}$. It is called a \textit{k-uniform m-cycle}.

\subsection{Known results}\label{sec21}
One important NP-completeness of $\pi$-\textsc{partition} is proved by Lov\'asz \cite{MR0363980} which is NP-completeness of 2-coloring problem of hypergraphs.
\begin{thm}[Lov\'asz \cite{MR0363980}]\label{thm0**0}
For a given hypergraph $G=(V, E)$, the $2$-coloring of $G$ is a function $f:V \rightarrow \{1, 2\}$ such that each edge of $G$ is not monochromatic. Deciding whether a given hypergraph $G$ has $2$-coloring is NP-complete even if $G$ is $3$-uniform hypergraph.
\end{thm}
Theorem \ref{thm0**0} shows the $(0, *, *, 0)$-\textsc{partition} is NP-complete.

To prove polynomial solvability results in section \ref{sec22}, we will look the theorem by Feder, Hell, Klein and Motwani \cite{MR2002173}. The original theorem was for graphs but the proof of the theorem also holds for hypergraphs.
\begin{thm}[Feder, Hell, Klein and Motwani \cite{MR2002173}]\label{thm1}
Let $\mathcal{S}$ and $\mathcal{D}$ be classes of hypergraphs closed under taking induced subgraph. Suppose there exist a constant $c$ such that every hypergraph $G$ contained in $\mathcal{S} \cap \mathcal{D}$ has at most $c$ vertices. Then, for every $n$-vertex hypergraph $G=(V,E)$, there are at most $n^{2c}$ partitions $(V_1, V_2)$ of $V$ such that $G[V_1] \in \mathcal{S}$ and $G[V_2] \in \mathcal{D}$. Furthermore, we can find all such partitions in $O(n^{2c+2}T(n))$ time where $T(n)$ is the time for recognizing $\mathcal{S}$ and $\mathcal{D}$. 
\end{thm}
This theorem directly shows the following corollary.
\begin{cor}\label{prop21}
For $\pi=(\pi_0,\pi_1,...,\pi_k) \in \{0, 1, * \}^{k+1}$, if $\pi_0=0$ and $\pi_k=1$, then the $\pi$-\textsc{partition} can be solved in polynomial time. Furthermore, there exist at most $n^{2k}$ $\pi$-partition of an input graph $G$ where $n$ is number of vertices of $G$ and we can find every partition in polynomial time.
\end{cor}
\begin{proof}
Take $\mathcal{S}$ be the class of empty hypergraph and $\mathcal{D}$ be the class of complete $k$-uniform hypergraph. In other word, $\mathcal{S}$ is the class of hypergraphs of the form of $G=(V, \emptyset)$ and $\mathcal{D}$ is class of hypergraph of the form of $G=(V, \binom{V}{k})$ where $\binom{V}{k}$ is all $k$-subsets of $V$. Then by applying Theorem \ref{thm1}, we can find every $(0, *, *, ..., *, 1)$-partition of an input $k$-uniform hypergraph $G$ in polynomial time. By finding every $(0, *, *, ..., *, 1)$-partition of an input $k$-uniform hypergraph $G$ and checking each partition is $\pi$-partition or not, we can find every $\pi$-partition of $G$.
\end{proof}

\section{Main results}

\subsection{Polynomial-time solvability}\label{sec22}
We will use Theorem \ref{thm1} and Corollary \ref{prop21} to prove more general fact. 
\begin{prop} \label{thm23}
If $\pi=(\pi_0, ..., \pi_k)$ contains both $1$ and $0$, then the $\pi$-\textsc{partition} is polynomial-time solvable. Furthermore, there exist at most $O(n^{3k})$ $\pi$-partitions for a $n$-vertex hypergraph $G$ and we can find every $\pi$-partition in polynomial-time.
\end{prop}
\begin{proof}
Suppose $\pi_0=0$ and $\pi_i=1$ for some $1 \leq i \leq k$.
Assume there exists a $\pi$-partition $(V_1, V_2)$ of an input hypergraph $G=(V, E)$. If $|V_1|\geq k-i$, then for any choice of $U_1 \subseteq V_1$ with $|U_1|=k-i$ and every subset $U_2 \subseteq V_1 \setminus U_1$ of size $i$, $U_1 \cup U_2 \notin E(G)$. Similarly, every subset $U_3 \subseteq V_2$ with $|U_3|=i$ satisfies $U_1 \cup U_3 \in E(G)$. Now, fix $U_1 \subseteq V_1$ and construct an $i$-uniform hypergraph $G_{U_1}$ as follows. (Note that $i$ maybe equal to 1.) 
\begin{center}
$V(G')=V(G) \setminus U_1$ and $U \in E(G')$ if and only if $U \cup U_1 \in E(G)$. 
\end{center}
Then by previous observation, $V_1 \setminus U_1$ is an independent set and $V_2$ is a clique.

Now, we will construct an algorithm. For every $U \subseteq V$ of size $k-i$, construct $G_U$ with respect to $U$. Use the algorithm from Theorem \ref{thm1} to find every $(1, *,...,*,0)$-partition of $G_U$. There are at most $n^{2i}$ $(1, *,...,*,0)$-partitions of $G_U$ and we can find every $(1, *,...,*,0)$-partition in polynomial time. For each $(1, *,...,*,0)$-partition $(V_1, V_2) $ of $V_{G_U}$, check $(V_1 \cup U , V_2)$ makes a $\pi$-partition of $G$ or not.
After check every $U \subseteq V(G)$ of size $k-i$, check every possibility that $V_1$ has size $<k-i$. By previous observation, this algorithm finds all $\pi$-partitions of $G$.
\end{proof}
Now, we will see one more polynomial-time solvable class of $\pi$-\textsc{partition}s. 
\begin{prop}\label{thm_alternating}
Let $\pi = (\pi_0,\pi_1,...,\pi_k) $. If $\pi_i=0$ for all even number $0 \leq i \leq k$ and $\pi_i=*$ for all odd number $0 \leq i \leq k$, then the $\pi$-\textsc{partition} is polynomial-time solvable.
\end{prop}
\begin{proof}
For a given $k$-uniform hypergraph $G=(V, E)$, we label the vertices as $u_1, u_2, ..., u_{|V|}$ and the edges as $e_1, e_2, ..., e_{|E|}$. Let $M$ be an $|E| \times |V|$ matrix in a Galois field $\mathds{F}_2$ where $M_{ij}=1$ if and only if $u_j \in e_i$. If $G$ has a $\pi$-partition $(V_1, V_2)$, then the linear equation $Mx=1$ in $\mathds{F}_2$ has a solution which is $x_i=1$ if and only if $u_i \in V_1$. Conversely, if the linear equation $Mx=1$ has a solution, then the vertex partition $(V_1, V_2)$ of $G$ where $u_i \in V_1$ if and only if $x_i=1$ is a $\pi$-partition of $G$. Since linear equation can be solved in polynomial time by applying Gaussian elimination, the $\pi$-\textsc{partition} is polynomial-time solvable.
\end{proof}

\subsection{NP-completeness of the $(0, *, 0, 0)$-\textsc{partition}}\label{sec23}
To prove Theorem \ref{thm3_uni}, we need one more NP-completeness theorem.
\begin{prop} \label{thm24}
The $(0, *, 0, 0)$-\textsc{partition} is NP-complete.
\end{prop}
\begin{proof}
Let $\pi=(0, *, 0, 0)$. First observe that the $3$-uniform hypergraph $H=(V, E)$ with $V=\{1,2,3,4\}$ and $E=\{\{1, 2, 4\}, \{1, 3, 4\}, \{2, 3, 4\} \}$ has the unique $\pi$-partition which is $V_1=\{4\}$ and $V_2=\{1, 2, 3\}$. Also, note that $C^3_6$ has three possible partitions $V_1=\{1, 4\}$, $V_1=\{2, 5\}$ and $V_1=\{3, 6\}$.

We will construct a polynomial-time reduction from the $3$-SAT to the $(0, *, 0, 0)$-\textsc{partition}. Let $\varphi(x_1, x_2,...,x_n)=\bigwedge_{j=1}^m(y_{j}^1 \vee y_{j}^2 \vee y_j^3)$ be an input 3-CNF formula where $y_j^k$ is one of $x_i$ or $\neg x_i$. If there exists a clause of the form of $(x_i \vee x_i \vee x_i)$, by removing clauses containing $x_i$ and removing all $\neg x_i$, we get a smaller 3-CNF formula which has same satisfiability. Therefore, we may assume that there is no clause consisting of three equal terms.

Let $G_\varphi$ be a $3$-uniform hypergraph with the vertex set $\bigcup_{i=1}^{n}\{x_i^1, x_i^2, \neg x_i^1, \neg x_i^2, u_i^1,$ $ u_i^2, u_i^3, u_i^4, u_i^5, u_i^6\} \cup \bigcup_{j=1}^{m} \{w_i^1, w_i^2, w_i^3, w_i^4, w_i^5, w_i^6\}$. If the $j$-th clause of $\varphi$ is of the form of $x_{j_1} \vee x_{j_1} \vee x_{j_2}$ and $j_1 \neq j_2$, let $E_j=\{\{x_{j_1}^1, w_j^1, w_j^2\}, \{{x}_{j_1}^2, w_j^3, w_j^4\},  \{x_{j_2}^1, w_j^5, w_j^6\},$ $ \{w_j^2, w_j^4, w_j^6\}\}$. Similarly, if the $j$-th clause of $\varphi$ is of the form of $x_{j_1} \vee x_{j_2} \vee x_{j_3}$ and $j_1, j_2, j_3$ are all distinct,  let $E_j=\{\{x_{j_1}^1, w_j^1, w_j^2\}, \{{x}_{j_2}^1, w_j^3, w_j^4\},  \{x_{j_3}^1, w_j^5, w_j^6\}, \{w_j^2, w_j^4, w_j^6\}\}$. If a clause contains $\neg x_i$, simply replace $x_i^1, x_i^2$ into $\neg x_i^1, \neg x_i^2$. The edge set of $G_\varphi$ is $\bigcup_{j=1}^{m} E_j \cup \bigcup_{i=1}^{n} \{\{x_{i}^1, \neg x_{i}^1, u_{i}^1\}, \{\neg x_{i}^1, u_{i}^1, x_{i}^2\},$ $ \{u_{i}^1, x_{i}^2, \neg x_{i}^2\},$ $ \{x_{i}^2, \neg x_{i}^2, u_{i}^2\},$ $\{\neg x_{i}^2, u_{i}^2,  x_{i}^1\}, \{u_{i}^2,  x_{i}^1, \neg x_{i}^1\},$  $ \{u_i^1, u_i^2, u_i^4 \},$ $ \{u_i^1, u_i^3, u_i^4\}, \{u_i^2, u_i^3, u_i^4\}\} $. Note that the subgraph of $G_\varphi$ induced by $\{x_{i}^1$, $\neg x_{i}^1$, $u_{i}^1$, $x_{i}^2$, $\neg x_{i}^2$, $u_{i}^2\}$ is isomorphic to the $C_6^3$. 

We claim that $G_\varphi$ has a $\pi$-partition if and only if $\varphi$ is satisfiable. Suppose $G_\varphi$ has a $\pi$-partition $(V_1, V_2)$. Then by previous observation, $u_i^4 \in V_2$ and $u_i^1, u_i^2, u_i^3 \in V_1$. Since   $\{x_{i}^1$, $\neg x_{i}^1$, $u_{i}^1$, $x_{i}^2$, $\neg x_{i}^2$, $u_{i}^2\}$ makes $C_6^3$, $x_i^1, x_i^2 \in V_1$ and $\neg x_i^1, \neg x_i^2 \in V_2$ or $x_i^1, x_i^2 \in V_2$ and $\neg x_i^1, \neg x_i^2 \in V_1$. For each $i$, assign true to $x_i$ if $x_i^1 \in V_2$ and assign false to $x_i$ if $x_i^2 \in V_2$. Suppose the $j$-th clause of $\varphi$ is $x_{j_1} \vee x_{j_2} \vee x_{j_3}$ ($j_i$ may equal). For each $E_j$, $\{w_j^2, w_j^4, w_j^6\} \in E_j$ so exactly one of $w_j^2, w_4^j, w_6^j$ is contained in $V_1$. This implies at least one of $x_{j_i}^k$ is contained in $V_2$. Therefore, at least one of $x_{j_i}$ has true value. If there exists $\neg x_{j_i}$ in a clause, by replacing $x_i^1, x_i^2$ into $\neg x_i^1, \neg x_i^2$, we get the same result. Therefore, assigned value of $x_i$ makes each clause true. So $\varphi$ is satisfiable. Conversely, suppose $\varphi$ is satisfiable. Fix any values of $x_i$'s that make $\varphi$ true. We will construct a $\pi$-partition $(V_1, V_2)$ of $G_\varphi$. If $x_i$ is true, make $x_i^1, x_i^2 \in V_2$ and $\neg x_i^1, \neg x_i^2 \in V_1$. If $x_i$ is false, make $x_i^1, x_i^2 \in V_1$ and $\neg x_i^1, \neg x_i^2 \in V_2$. Suppose the $j$-th clause of $\varphi$ is $x_{j_1} \vee x_{j_2} \vee x_{j_3}$ ($j_i$ may equal). Choose one $i$ such that $x_{j_i}$ is true. Make $w_{2i} \in V_1$, $w_{2i-1} \in V_2$ ,$w_{2k}\in V_2$ for $k \neq j$ and $w_{2k-1}\in V_2$ if $x_{j_k}$ is false, $w_{2k-1}\in V_1$ if $x_{j_k}$ is true for $k \neq j$. Finally, make $u_i^4 \in V_2$ and $u_i^1, u_i^2, u_i^3 \in V_1$. Then we get it is the $\pi$-partition of $G_\varphi$. This process can be done in time polynomial of input size so it gives a  polynomial-time reduction from the 3-SAT to the $\pi$-\textsc{partition}.
\end{proof}
We can use this theorem to prove the fact that the exact cover problem is NP-complete even if $|\{ Y \in S \mid x \in Y \}|=3$ for all $x \in X$. The exact cover problem is a decision problem for a given pair $(X,C)$ where $X$ is a finite set and $C$ is a subset of power set of $X$, deciding there exist a subset $C' \subseteq C$ such that each element in $X$ lies in exactly one member of $C'$. It is known as NP-complete \cite{MR519066}. For a given $3$-uniform hypergraph $G=(V, E)$, let $X=E$ and $S=\{$set of edges containing $v \mid v \in V\}$, we can easily see that the exact cover problem under the condition and the $(0, *, 0, 0)$-\textsc{partition} are polynomial-time equivalent.

Now we will prove Theorem \ref{thm3_uni}.
\begin{proof}[Proof of Theorem \ref{thm3_uni}]
By Proposition \ref{thm23} and considering trivial cases, we can conclude that $\pi$-\textsc{partition} of $3$-uniform hypergraph is polynomial-time solvable except $\pi=(0, *, *, 0)$ ,$\pi=(1, *, *, 1)$, $\pi=(0, 0, *, 0)$, $\pi=(0, *, 0, 0)$, $\pi=(1, *, 1, 1)$ or $\pi=(1, 1, *, 1)$ for $\pi \in \{0, 1, *\}^{4}$. The $(0, *, *, 0)$-\textsc{partition} is NP-complete by Theorem \ref{thm0**0} and the $(0, *, 0, 0)$-\textsc{partition} is NP-complete by Proposition \ref{thm24}. By taking complement and by Lemma \ref{lem1}, the $(1, *, *, 1)$-\textsc{partition} is polynomial-time equivalent to the $(0, *, *, 0)$-\textsc{partition} and the $(1, 1, *, 1)$-\textsc{partition}, the $(1, *, 1, 1)$-\textsc{partition} and the $(0, 0, *, 0)$-\textsc{partition} are polynomial-time equivalent to the $(0, *, 0, 0)$-\textsc{partition}.
\end{proof}
\subsection{Polynomial-time reductions}\label{sec24}

In this subsection, we will prove NP-completeness using polynomial-time reduction. Note that by Theorem \ref{thm23} and by taking complement, we may assume that $\pi$ does not contains $1$. We say $\pi$ is \textit{$1$-free} if $\pi$ does not contains $1$.

First we define a map $\sigma:\{0, *\}^{k+1} \rightarrow \{0, *\}^{k+2}$ as $\sigma(\pi)_0=\pi_0, \sigma(\pi)_{k+1}=\pi_k$ and $\sigma(\pi)_i=*$ if $\pi_i=*$ or $\pi_{i-1}=*$ and $\sigma(\pi)_i=0$ otherwise for $1 \leq i \leq k$. Then we get the following proposition.
\begin{prop} \label{thm25}
If there are no $i$ such that $\pi_i=\pi_{i+2}=*$ and $\pi_{i+1}=0$, there exist a polynomial-time reduction from the $\pi$-\textsc{partition} to the $\sigma(\pi)$-\textsc{partition}.
\end{prop}
\begin{proof}
For a given $k$-uniform hypergraph $G=(V,E)$, we construct a $(k+1)$-uniform hypergraph $G'=(V', E')$ with the vertex set
$V'=V \cup \{u_1, u_2, ..., u_{k+1} \}$ and the edge set $E'=\{ e \cup \{u_i\} \mid 1 \leq i \leq k+1, e \in E \} \cup \{u_1, ..., u_{k+1}\} $. If $G$ has a $\pi$-partition $(V_1, V_2)$, then the partition $(V_1 \cup \{u_1\}, V_2 \cup \{u_2,...,u_{k+1}\})$ is a $\sigma(\pi)$-partition of $G'$. Conversely, if $G'$ has a $\sigma(\pi)$-partition $(V_1, V_2)$, there are at least one pair of integers $(i, j)$ satisfying $u_i \in V_1$ and $u_j \in V_2$. Therefore, for each $\{w_1, w_2, ...,w_k\} \in E$, $|\{w_1, w_2, ...,w_k\} \cap V_1|=i$ implies that $\sigma(\pi)_i=\sigma(\pi)_{i+1}=*$. By assumption, it implies that $\pi_i=*$ so $(V_1 \setminus \{u_1, ..., u_{k+1}\}, V_2 \setminus \{u_1, ..., u_{k+1}\})$ is a $\pi$-partition of $G$. 
\end{proof}
We can use this proposition and Theorem \ref{thm0**0} to prove that $2$-coloring problem of $k$-uniform hypergraph is NP-complete for all $k \geq 3$.

Note that number of 0's in $\sigma(\pi)$ is smaller than $\pi$. On the other hand, the next proposition produce another reduction from the $\pi$-\textsc{partition} to the $\pi'$-\textsc{partition} such that number of 0's in $\pi'$ is strictly larger than number of 0's in $\pi$. 
\begin{lem}
Suppose $\pi=(\pi_0, ..., \pi_k)$ and $\pi'=({\pi'}_0, ..., {\pi'}_{k'})$ are $1$-free and there are no consecutive *'s in $\pi'$. If there exist nonnegative integers $j_1, j_2, ..., j_k$ such that ${\pi'}_{\sum_{m \in I} j_m}=*$ if $\pi_{|I|}=*$ and ${\pi'}_{\sum_{m \in I} j_m}=0$ if  $\pi_{|I|}=0$ for all $I \subseteq \{1, 2, ..., k\}$, then there exist a polynomial-time reduction from the $\pi$-\textsc{partition} to the $\pi'$-\textsc{partition}.
\end{lem}
\begin{proof}
If $\pi'$ is a zero vector, then $\pi$ is also a zero vector so clearly there are polynomial-time reduction. Therefore, we may assume that $\pi'$ is not a zero vector.  

For a given $k$-uniform hypergraph $G=(V,E)$, construct a $k'$-uniform hypergraph $G'$ as follows. 
\begin{align*}
V' &= \{u^i \mid u \in V, 1 \leq i \leq k' \} \cup \{w_u^i \mid u \in V, 1 \leq i \leq k'-1\} \\
E' &= \{\{w_u^1, w_u^2, ..., w_u^{k'-1}, u^i\} \mid u \in V, 1 \leq i \leq k'\}\\ &\cup  \{\{u_1^1, u_1^2, ..., u_1^{j_1}, u_2^1, ..., u_2^{j_2}, ..., u_k^{j_k}\} \mid \{u_1, u_2, ..., u_k\} \in E \}
\end{align*}
Suppose $G'$ has a $\pi'$-partition $(V'_1, V'_2)$. If $u^1 \in V'_1$, then $u^i \in V'_1$ for all $i$ because if not, $|\{w_u^1, w_u^2, ..., w_u^{k'-1}, u^1\} \cup V'_1|=|\{w_u^1, w_u^2, ..., w_u^{k'-1}, u^i\} \cup V'_1|+1$. It is contradicting to the assumption. By the same argument, if $u^1 \in V'_2$, then $u^i \in V'_2$ for all $i$ Let $({V}_1, {V}_2)$ be a partition of $V$ where $u \in {V}_1$ if and only if $u^1 \in V'_1$. For every edge $\{u_1, u_2, ..., u_k\} \in E$, if $|\{u_1, u_2, ..., u_k\} \cap V_1| = s$, then there exists $I \subseteq \{1, 2, ..., k\}$ such that $|I|=s$ and ${\pi'}_{\sum_{m \in I} j_m}=*$. It implies $\pi_i=*$ because if not, then ${\pi'}_{\sum_{m \in I} j_m}=0$ for all $I \subseteq \{1, 2, ..., k\}$ with $|I|=s$. Therefore, $(V_1, V_2)$ is a $\pi$-partition of $G$. Conversely, assume $G$ has a $\pi$-partition $(V_1, V_2)$. We choose any $m$ such that $\pi'_m=*$. Let $(V'_1, V'_2)$ as $u^i \in V'_1$ if and only if $u \in V_1$, $w_u^i \in V'_1$ for $1 \leq i \leq m-1$, $w_u^i \in V'_2$ for $m+1 \leq i \leq k'-1$ and $w_u^m \in V'_1$ if and only if $u \in V_2$. Then it is a $\pi'$-partition of $G'$ since $\pi_i=*$ implies ${\pi'}_{\sum_{m \in I} j_m}=*$ for all $I \subseteq \{1, 2, .., k\}$ with $|I|=i$ and $|\{w_u^1, w_u^2, ..., w_u^{k'-1}, u^i\} \cap V'_1|=m$ for any $1 \leq i \leq k', u \in V$.
\end{proof}
It is hard to check whether such $j_1, j_2, ..., j_k$ exist or not. However, by taking all $j_m=2$, we get the following useful proposition.
\begin{prop}\label{cor1}
Suppose $\pi=(\pi_0, ..., \pi_k)$ is 1-free. Then there exist a polynomial-time reduction from the $\pi$-\textsc{partition} to the $(\pi_0, 0, \pi_1, 0, \pi_2, 0, ..., 0, \pi_k)$-\textsc{partition}.
\end{prop}
We observe that no consecutive *'s condition is only needed for making $u^1, u^2, ..., u^{k'}$ belong to the same side. Therefore, for a vector $\pi'$, if there exist a hypergraph $H$ and vertices $u_1, u_2 \in V(H)$ such that it has at least one $\pi'$-partition and for every $\pi'$-partition of $H$, $u_1$ and $u_2$ lie in same side, then the lemma also holds without no consecutive *'s condition. This observation also holds for the next proposition. On the other hand, since $(\pi_0, 0, \pi_1, 0, ..., 0, \pi_k)$ has no consecutive *'s, Proposition \ref{cor1} make it easy to apply Lemma \ref{thm25} and the next proposition. 
\begin{prop}\label{thm_add0}
Suppose $\pi=(\pi_0, \pi_1, ..., \pi_k)$ is 1-free and the has no consecutive *'s. Then there exists a polynomial-time reduction from the $\pi$-\textsc{partition} to the $(0, \pi_0,  \pi_1,  ...,  \pi_k)$-\textsc{partition}.
Furthermore, if $(0, \pi_0, \pi_1, ..., \pi_k)=(\pi_k, \pi_{k-1}, ..., \pi_0, 0)$, then it also holds without no consecutive * condition. 
\end{prop}
\begin{proof}
If $(0, \pi_0, \pi_1, ..., \pi_k) = (\pi_k, \pi_{k-1}, ..., \pi_0, 0)$, for a given $k$-uniform hypergraph $G=(V, E)$, we construct a $k+1$-uniform hypergraph $G'=(V', E')$ with the vertex set $V'=V \cup \{u\}$ and the edge set $E'=\{e \cup u \mid e \in E\}$. If $G$ has a $\pi$-partition $(V_1, V_2)$, then $(V_1 \cup \{u\}, v_2)$ is a $(0, \pi_0, \pi_1, ..., \pi_k)$-partition of $G'$. Conversely, suppose $G'$ has a $(0, \pi_0, \pi_1, ..., \pi_k)$-partition $(V'_1, V'_2)$. If $u \in V'_1$, then $(V'_1 \setminus \{u\}, V'_2)$ is a $\pi$-partition of $G$. If $u \in V'_2$, then $(V'_1 , V'_2\setminus \{u\})$ is a $(0, \pi_0, \pi_1, ...,\pi_{k-1})$-partition of $G$ so $(V'_2\setminus \{u\}, V'_1)$ is a $(\pi_{k-1}, ..., \pi_0, 0)$-partition of $G$ and $(\pi_{k-1}, ..., \pi_0, 0)=\pi$.

If $(0, \pi_0, \pi_1, ..., \pi_k) \neq (\pi_k, \pi_{k-1}, ..., \pi_0, 0)$, we choose $m$ such that $\pi_m=*$ but $\pi_{k-m-1}=0$. We define a $(k+1)$-uniform hypergraph $G'=(V', E')$ with the vertex set $V'=V \cup \{u_1, ..., u_{2k}, w_1, ..., w_{k+1}\}$ and the edge set $E'=\{ \{u_1, u_2,...,u_{k}, w_i \} \mid 1 \leq i \leq m+1\} \cup \{ \{u_{k+1}, u_{k+3},...,u_{2k}, w_i \} \mid m+2 \leq i \leq k+2\} \cup \{ e \cup \{ w_1\} \mid e \in E\} \cup \{\{w_1, w_2,...,w_{k+1}\}\}$. 

Suppose $G'$ has a $(0, \pi_0,  \pi_1,  ...,  \pi_k)$-partition $(V'_1, V'_2)$. Then $w_1, w_2, ..., w_{m+1}$ belongs to the same part because if not, $|\{u_1, u_2,...,u_{k}, w_i \} \cup V_1|=|\{u_1, u_2,...,u_{k}, w_j \}|+1$ for some $1 \leq i, j \leq m+1$ and it is contradiction to $\pi$ has no consecutive *'s. By the same reason, $w_{m+2}, ..., w_{k+1}$ belongs to the same part. Since $\pi_m=*$ but $\pi_{k-m-1}=0$, $w_i \in V'_1$ for $1 \leq i \leq m+1$ and $w_i \in V'_2$ for $m+2 \leq i \leq k+1$. Therefore, $|\{v_1, ...,v_k, w_1\} \cap V'_1|=i$ if and only if $|\{v_1, ..., v_k\} \cap V'_1|=i-1$. It shows $(V'_1 \cap V, V'_2 \cap V)$ is a $\pi$-partition of $G$. Conversely, if $G$ has a $\pi$-partition $(V_1, V_2)$, then $(V_1 \cup \{w_1, ..., w_{m+1}\} \cup \{u_1, ..., u_m\} \cup \{u_{k+1}, ..., u_{k+m+1}\}, V_2 \cup \{w_{m+2}, ..., w_{k+1}\} \cup \{u_{m+1}, ..., u_{k}\} \cup \{u_{k+m+2}, ..., u_{2k}\})$ is a $(0,$ $ \pi_0,  \pi_1,  ...,  \pi_k)$-partition of $G'$.

\end{proof}
\begin{cor}\label{cor2}
Let $k \geq 3$ be an integer and $\pi \in \{0, 1, *\}^{k+1}$ is $1$-free.
If $\pi$ contains exactly one * and $\pi_0=\pi_k=0$, then the $\pi$-\textsc{partition} is NP-complete.
\end{cor}
\begin{proof}
By Proposition \ref{thm24}, Proposition \ref{thm_add0} and Lemma \ref{lem1}, it is clear.
\end{proof}
Now, we will prove Theorem \ref{thm4_uni} and Theorem \ref{thm5_uni}.
\begin{proof}[Proof of Theorem \ref{thm4_uni}]
By Theorem \ref{thm0**0}, Proposition \ref{thm24} and Proposition \ref{thm25}, the $(0, *, *, *, 0)$-\textsc{partition}, the $(0, 0, *, *, 0)$-\textsc{partition} and the $(0, *, *, 0, 0)$-\textsc{partition} are NP-complete. By Corollary \ref{cor2}, the $(0, *, 0, 0, 0)$-\textsc{partition}, the $(0, 0, *, 0, 0)$-\textsc{partition} and the $(0, 0,$ $ 0, *, 0)$-\textsc{partition} are NP-complete. By Proposition \ref{thm_alternating}, the $(0, *, 0, *, 0)$-\textsc{partition} is polynomial-time solvable. By Proposition \ref{thm23} and considering trivial cases, remaining cases are polynomial-time solvable. By combining these results, we get the proof of Theorem \ref{thm4_uni}.
\end{proof}
\begin{proof}[Proof of theorem \ref{thm5_uni}]
By applying Proposition \ref{thm24} and Proposition \ref{thm25} to result of Theorem \ref{thm4_uni}, we get all NP-completeness of Theorem \ref{thm5_uni}. Remaining cases are polynomial-time solvable by Proposition \ref{thm24} and by considering trivial cases. It proves Theorem \ref{thm5_uni}.
\end{proof}

\section*{Acknowledgement}

The author would like to thank Jaehyun Koo for pointing out the partitioning problem of a split graph and thank Prof.~Sang-il Oum for very helpful advice. 

\bibliographystyle{plain}
\bibliography{cites}

\begin{thebibliography}{1}

\bibitem{MR2002173}
Tomas Feder, Pavol Hell, Sulamita Klein, and Rajeev Motwani.
\newblock List partitions.
\newblock {\em SIAM J. Discrete Math.}, 16(3):449--478, 2003.

\bibitem{MR519066}
Michael~R. Garey and David~S. Johnson.
\newblock {\em Computers and intractability}.
\newblock W. H. Freeman and Co., San Francisco, Calif., 1979.
\newblock A guide to the theory of NP-completeness, A Series of Books in the
  Mathematical Sciences.

\bibitem{zbMATH03773637}
Peter~L. {Hammer} and Bruno {Simeone}.
\newblock {The splittance of a graph.}
\newblock {\em {Combinatorica}}, 1:275--284, 1981.

\bibitem{MR0363980}
L.~Lov\'{a}sz.
\newblock Coverings and coloring of hypergraphs.
\newblock In {\em Proceedings of the {F}ourth {S}outheastern {C}onference on
  {C}ombinatorics, {G}raph {T}heory, and {C}omputing ({F}lorida {A}tlantic
  {U}niv., {B}oca {R}aton, {F}la., 1973)}, pages 3--12. Utilitas Math.,
  Winnipeg, Man., 1973.

\bibitem{MR2424844}
\'Ad\'am Tim\'ar.
\newblock Split hypergraphs.
\newblock {\em SIAM J. Discrete Math.}, 22(3):1155--1163, 2008.

\bibitem{MR3734241}
Dmitriy Zhuk.
\newblock A proof of {CSP} dichotomy conjecture.
\newblock In {\em 58th {A}nnual {IEEE} {S}ymposium on {F}oundations of
  {C}omputer {S}cience---{FOCS} 2017}, pages 331--342. IEEE Computer Soc., Los
  Alamitos, CA, 2017.

\end{thebibliography}

\end{document}